\documentclass[a4paper, 11pt,reqno]{amsart}

\usepackage{amsmath,amsthm,amsfonts,amscd,amssymb}
\usepackage{amssymb}
\usepackage{url}
\usepackage{graphicx}
\usepackage{enumerate}
\usepackage{hyperref}

\renewcommand{\leq}{\leqslant}
\renewcommand{\geq}{\geqslant}

\newcommand{\N}{\mathbb{N}}
\newcommand{\Z}{\mathbb{Z}}

\newcommand{\R}{\mathbb{R}}
\newcommand{\C}{\mathbb{C}}

\DeclareMathOperator{\trace}{Tr}

\DeclareMathOperator{\rk}{rk}

\renewcommand{\phi}{\varphi}

\newtheorem{theorem}{Theorem}[section]
\newtheorem{definition}[theorem]{Definition}
\newtheorem{proposition}[theorem]{Proposition}
\newtheorem{remark}[theorem]{Remark}
\newtheorem{lemma}[theorem]{Lemma}

\newtheorem{corollary}[theorem]{Corollary}

\def\be{\begin{eqnarray}}
\def\ee{\end{eqnarray}}

\begin{document}

\title{On the convergence of output sets of quantum channels}

\author{Beno\^{\i}t Collins}
\address{
D\'epartement de Math\'ematique et Statistique, Universit\'e d'Ottawa,
585 King Edward, Ottawa, ON, K1N6N5 Canada,
WPI Advanced Institute for Materials Research Tohoku University, Mathematics Unit 
2-1-1 Katahira, Aoba-ku, Sendai, 980-8577 Japan
and 
CNRS, Institut Camille Jordan Universit\'e  Lyon 1,
France} 
\email{bcollins@uottawa.ca}
\author{Motohisa Fukuda}
\address{
Zentrum Mathematik, M5, Technische Universit\"at M\"unchen, Boltzmannstrasse 3, 85748 Garching, Germany}
\email{m.fukuda@tum.de}
\author{Ion Nechita}
\address{
CNRS, Laboratoire de Physique Th\'eorique, IRSAMC, Universit\'e de Toulouse, UPS, 31062 Toulouse, France}
\email{nechita@irsamc.ups-tlse.fr}

\begin{abstract}
We study the asymptotic behavior of the output states of sequences of quantum channels.
Under a natural assumption, we show that the output set converges to a compact convex set, clarifying and substantially generalizing
results in \cite{bcn}.
Random mixed unitary channels satisfy the assumption; we give a formula for the asymptotic maximum output infinity norm and 
we show that  the minimum output entropy and the Holevo capacity have a simple relation for the complementary channels.
We also give non-trivial examples of sequences  $\Phi_n$ such that 
along with any other quantum channel  $\Xi$, we have convergence of the output set of $\Phi_n$ and 
$\Phi_n\otimes \Xi$ simultaneously; the case when $\Xi$ is entanglement breaking is investigated in details.

\end{abstract}

\maketitle

\tableofcontents
 
\section{Introduction}
Quantum channels are of central importance in Quantum Information Theory, and in the meantime, 
many mathematical quantities that are 
associated to quantum channels are still not very well understood. This is the case, for example, for 
the maximum output infinity norm, the minimum output entropy and Holevo capacity. 
These quantities as well as other important quantities turn out to actually  depend only on the image of 
the collection all possible output states. 
Incidentally, the output set - a compact convex subset in the set of all state - turns out to be an interesting
geometric object that has nice interpretations in the theory of entanglement, statistics, 
free probability and others.
However, identifying the output set for a given channel turns out to be a difficult task.
So, instead, we analyze sequences of quantum channels which have nice asymptotic properties. 
Recently, research on random quantum channels in terms of eigenvalues has lead to important advances 
in the understanding of quantum channels, in particular in relation to the problem of additivity of the minimum output
entropy, see for example \cite{cn2}, \cite{bcn} and \cite{bcn13}.

In this paper we elaborate an axiomatic and systematic study of properties of sequences of quantum channels
that ensure the convergence of the output set towards a limit.
It turns out that the sufficient conditions that we unveil are not of random nature, although all examples available
so far rely on random constructions. 
The main results of this paper are Theorems \ref{thm:LA1} and  \ref{thm:LA2}. The idea underlying these
theorems was already available in \cite{bcn}, but  we considerably simplify and conceptualize the argument,
and we remove all probabilistic considerations from our main argument. We start with examples of deterministic quantum channels (or projections) which fit our axiomatic framework.
Then, we treat random quantum channels and random mixed unitary channels as examples of this axiomatic approach. 
To do so, we rely on recent results of \cite{cm} and \cite{HT}
where the strong asymptotic freeness of Haar unitary matrices and constant matrices or 
extension of strong convergence to polynomials with matrix coefficients is proved. 

Our paper is organised as follows.
Section \ref{sec:image}
contains definitions and reminders about quantum channel and quantities associated to them.
Then, our main result is stated and proved in
Section  \ref{sec:linear-alg}. 
Then, in Section \ref{sec:asymptotic}
  our main result is applied to the convergence of entropies. 
Section \ref{sec:free-prob} introduces some results from random matrix theory and free probability. A subclass of entanglement-breaking channels is investigated in Section \ref{sec:non-random},
then we discuss
in Section \ref{sec:examples}  examples of our main result: random Stinespring channels 
(subsection \ref{sec:examples-Stinespring}) and random mixed unitary channels (subsection \ref{sec:examples-RMUC}). 
Finally, in Section \ref{sec:tensor}, we discuss tensor products of channels, and especially entanglement-breaking channels are discussed in details.

\section{Quantum channels and their image}
\label{sec:image}
\subsection{Notation}
Following the quantum information theoretic notation, we call \emph{quantum states} semidefinite positive matrices of unit trace 
\begin{equation}
D_k = \{ A \in M_k \, : \, A \geq 0 \text{ and } \mathrm{Tr} A = 1\},
\end{equation}
where we write $M_k = M_k(\mathbb C)$.
A \emph{quantum channel} is a completely positive and 
 trace preserving linear map $\Phi :M_N\to M_k$.
Following the Stinespring's picture \cite{Stinespring1955}, we view any quantum channel $\Phi$ as an isometric embedding 
of $\mathbb C^N$ into $\mathbb C^k \otimes \mathbb C^n$, to which we apply a partial trace.
\be
V : \mathbb C^N \rightarrow \mathbb C^k \otimes \mathbb C^n
\ee
That is, 
\be\label{eq:channel-Stinespring}
\Phi =({\rm id}_k\otimes Tr_n) \circ E
\ee
where $E (\cdot) = V \cdot V^*$ is a non-unital embedding of $M_N$ in $M_k\otimes M_n$.
We define the complementary channel of $\Phi$ by
\be
 \tilde \Phi = ( Tr_k \otimes{\rm id}_n)\circ E
\ee
Note that for a pure input, its outputs via $\Phi$ and $\tilde \Phi$ share the same non-zero eigenvalues,
but it is not the case in general for a mixed input.  We also define the adjoint channel $\Phi^*$, which is the adjoint of $\Phi$ with respect to the Hilbert-Schmidt scalar product in $M_k$:
\begin{equation}\label{def:adjoint}
\mathrm{Tr}[\Phi(X)^* Y] = \mathrm{Tr}[X^* \Phi^*(Y)].
\end{equation}
If $\Phi$ is defined via a Stinespring dilation using an isometry $V$ as in \eqref{eq:channel-Stinespring}, then 
\begin{equation}
\Phi^*(Y) = V^* (Y \otimes I_n) V.
\end{equation}

In Quantum Information language, $\mathbb C^n$ is called the environment.
In this paper, we are interested in a sequence of quantum channels, that we will index by the 
environment $n$,  $\Phi_n:M_N\to M_k$.
From now on, our setting is as follows:
$k,n,N \in \mathbb N$ are such that $k$ is fixed  and $N\in \mathbb N$ is any function of $n\in \mathbb N$.
Importantly, a quantum channel is defined,
up to a  unitary conjugation on the input,
 by $P_n = V_nV_n^*$, 
which is  the unit of $M_N$ embedded in $M_k\otimes M_n$.

Let $D_N$ be the collection of states in $M_N$, and $D^p_N \subset D_N$ be collection of extremal (pure) states,
i.e. rank-one (self-adjoint) projections. 
We are interested in $L_n=\Phi_n(D^p_N)$, 
which is the image of all the pure states under the quantum channel $\Phi_n$. 
One can see that $L_n$ is a compact subset of $D_k$, but
not always convex, although so is $K_n=\Phi_n(D_N)$. 

The task to classify all the possible sets $K_n$ and $L_n$ arising from this construction seems
to be out of reach. Instead, we focus our attention on possible asymptotic behaviours of 
$K_n$ and $L_n$ as $n \rightarrow \infty$.

In Section \ref{sec:linear-alg}, we identify some assumption
with which $K_n$ and $L_n$ converge to some well-described compact convex set as $n \rightarrow \infty$.  
Then, we present examples of sequences of random projections $\{P_n\}_{n \in \mathbb N}$ 
which satisfy this assumption with probability one.

\subsection{Entropies and Capacities} 
We introduce three quantities associated with quantum channels.
 
Firstly, the maximum output infinity norm of channel $\Phi$ is defined as 
\be
\| \Phi\|_{1,\infty} = \max_{\rho \in D_N} \|\Phi(\rho)\|_{\infty}
\ee
where $1$ and $\infty$ represent norms used for the input and output spaces respectively. 

Secondly, the minimal output entropy (MOE) of channel $\Phi$ is defined as
\be
S^{\min} (\Phi) = \min_{\rho\in D_N} S(\Phi(\rho)) 
\ee 
Here, $S(\cdot)$ is the von Neumann entropy. 

Thirdly, 
the Holevo capacity (HC) of channel $\Phi$ is defined as 
\be
\chi(\Phi) = \max_{\{p_i,\rho_i\}} S\left(\Phi\left(\hat \rho \right)\right) - \sum_i p_iS\left(\Phi( \rho_i)\right) 
\ee
Here, $\{p_i\}_i$ is a probability distribution, $\{\rho_i\}_i \subset D_N$ and $\hat \rho = \sum_i p_i \rho_i$.
Note that 
\be
S^{\min} (\tilde \Phi) = S^{\min} (\Phi) \label{MOE1}
\ee
but 
\be
\chi (\tilde \Phi) \not = \chi (\Phi) \label{HC1}
\ee
in general. 
 
It is a rather direct observation that $S_{\min}$ and $\chi$ depend only on the output of the channel.
Therefore, for any convex set $M \subset D_k$, 
it is natural to define
\begin{align}
S^{\min}(M) &= \min_{X \in M} S(X) \label{MOE2}\\
\chi(M) &= \max_{\{p_i, X_i\}} S\left(\sum_i p_i X_i\right) - \sum_i p_iS\left(X_i\right).\label{HC2}
\end{align}
where $X_i \in M$. 

For those quantities  
one can think of additivity questions:  
\be
\chi(\Phi \otimes \Omega) \overset ? = \chi(\Phi) + \chi(\Omega) \\
S^{\min} (\Phi \otimes \Omega)\overset ? = S^{\min} (\Phi ) + S^{\min} (\Omega)
\ee 
for two quantum channels. 
These equalities are not true in general; 
additivity of MOE was disproved by Hastings \cite{hastings} and
this non-additivity can be translated to be the one for HC \cite{shor}.  
In terms of information theory, the additivity of HC is important.
Suppose in particular that 
\be
\chi\left(\Phi^{\otimes r}\right) = r \chi(\Phi)
\ee 
for some channel $\Phi$, then the classical capacity over this channel $\Phi$
has a one-shot formula:
\be
\lim_{r \to \infty} \frac 1 r \chi\left(\Phi^{\otimes r}\right)  = \chi(\Phi)
\ee 
Additivity of MOE itself is also interesting because it measures purity of channels,
but caught more attention when Shor proved the equivalence between the two additivity questions \cite{shor}. 
Moreover, a breakthrough was made by disproving additivity of MOE \cite{hastings}
with a use of random matrix theory as MOE concerns eigenvalues of matrices
whereas HC depends on the geometry of output states. 
By contrast, our paper sheds light on not only MOE but also HC
because we consider geometry of output states (at least, of single channels).

\section{Main result - linear algebra and convex analysis}\label{sec:linear-alg}  
\subsection{Preliminary} 
For a sequence of sets $\{S\}_{n\in \N}$ we use the following standard notations
of lim-inf and lim-sup:
\be
\varliminf_{n \rightarrow \infty} S_n = \bigcup_{N\in\N} \bigcap_{n \geq N} S_n 
\qquad \text{and} \qquad
\varlimsup_{n \rightarrow \infty} S_n = \bigcap_{N\in\N} \bigcup_{n \geq N} S_n 
\ee
If $S$ is a subset of a topological space,
we denote the \emph{interior} of $S$ by $S^\circ$  and the \emph{closure} $S^{cl}$. 
Suppose we have a \emph{convex} set $K$ in a real vector space;
 any line segment joining two points of $K$ is included in $K$.
Then, we have the following two definitions:  
\begin{enumerate}[(a)]
\item A point $x\in K$ is an  \emph{extreme point} if 
$x$ does not lie in any open line segment joining two points of $K$. 
\item A point $x \in K$ is an \emph{exposed point} if there exists a  supporting hyperplane
which intersects with $K$ only at one point.  
\end{enumerate}

We also use the following notation:
\be
{\rm hull} (\{x_i\}_{i=1}^m) = 
\left\{ \sum_{i=1}^m \lambda_i x_i :  \quad \lambda_i \geq 0 \quad \text{and}\quad  \sum_{i=1}^m \lambda_i =1  \right\}
\ee
which is called the \emph{convex hull} of $\{x_i\}_{i=1}^m$. 

\begin{theorem}[Steinitz \cite{Steinitz1913}]\label{Steinitz}
Suppose we have a convex and compact set $K \subset \mathbb R^d$. Then, for any interior point of $K$ 
we can choose at most $2d$ extreme points of $K$ whose convex hull includes the point 
within the interior. 
\end{theorem} 

\begin{theorem}[Straszewicz \cite{Straszewicz1935}] \label{Straszewicz} 
For any closed convex set $K$,
the set of exposed points is dense in the set of extreme points.   
\end{theorem} 

We define conditions which ensure the limiting convex set of output states. 
\begin{definition}\label{def:condition}
A sequence of projections $P_n \subset  M_{kn}$ is said to satisfy 
the condition $\mathcal C_m$ if 
for all $A \in D_k$,
the following $m$ infinite sequences in $n \in \N$:
\be\label{as2}
\lambda_{1} (P_n (A \otimes I_n)P_n), \ldots, \lambda_{m} (P_n (A \otimes I_n)P_n)  
\ee 
converge to a common limit, which we denote $f(A)$. 
Here, $\lambda_{i}(\cdot)$ is the $i$-th largest eigenvalue. 
Note that $\mathcal C_m \implies \mathcal C_l$ for $l < m$.  
\end{definition}  

Let $(\Phi_n)$ be a sequence of quantum channels associated with the sequence of projections $(P_n)$, that is $\Phi_n$ is 
defined by \eqref{eq:channel-Stinespring} and $P_n = V_nV_n^*$. One can then define the function $f$ in terms of the adjoint 
channels $\Phi_n^*$ defined in \eqref{def:adjoint}:
\begin{proposition}\label{prop:ad}
Let $\Phi:M_N \to M_k$ be a quantum channel defined by an isometry $V:\mathbb C^N \to \mathbb C^k \otimes \mathbb C^n$, and 
put $P = V V^*$. Then, for any $A \in D_k$, the non-zero eigenvalues of the matrices $P(A \otimes I_n)P$ and $\Phi^*(A)$ are identical. 
In particular, a sequence of projections $(P_n)$ satisfies condition $\mathcal C_m$ if and only if, for all $A \in D_k$, the $m$ largest 
eigenvalues of $\Phi_n^* (A)$ converge to $f(A)$.
\end{proposition}
\begin{proof}
We have
\begin{equation}
P(A \otimes I_n)P = VV^*(A \otimes I_n)VV^* = V \Phi^*(A) V^*.
\end{equation}
The conclusion follows from the fact that $V$ is an isometry, hence the matrices $V\Phi^*(A)V^*$ and $\Phi^*(A)$ have 
the same non-zero eigenvalues.
\end{proof}

Les us start by recording an obvious upper bound:
\begin{lemma}\label{lem:upperbound}
If a sequence of projections $P_n$ satisfies the condition $\mathcal C_1$,
then for any sequence $x_n \in \mathbb C^k \otimes \mathbb C^n$ such that $x_nx_n^* \leq P_n$ 
we have
\be
\varlimsup_{n \rightarrow \infty} \trace [X_n A] \leq f(A) \qquad \forall A \in D_k 
\ee
where $X_n = \trace_{\mathbb C^n} [x_n x_n^*]$.  
 \end{lemma}
\begin{proof} For all $n \in \mathbb N$, 
\[
\trace \left[X_n A \right] = \trace\left[x_{n} x_{n}^* (A\otimes I_{n})|\right] 
\leq \max_{vv^*\leq  P_{n}} \langle v, (A\otimes I_{n}) v \rangle = \|P_{n} (A\otimes I_{n})P_{n} \|_\infty   \to f(A)
\]
\end{proof} 

Next, we prove existence of sequence of optimal vectors: 
\begin{lemma}\label{lem:condition-subspace} 
If a sequence of projections $P_n \subset  M_{kn}$  satisfies 
the condition $\mathcal C_m$, then for any $A \in D_k$ there exists a sequence of $m$-dimensional subspaces 
$W_n \subset \mathrm{range}\, P_n$ 
for large enough $n \in \mathbb N$ such that  
any sequence of unit vectors $x_n \in W_n$ satisfies 
\be \label{eq:condition-subspace}
\trace [X_n A ] \to f(A)  \quad \text{as } n \to \infty
\ee 
Here, $X_n = \trace_{\C^n} \left[x_nx_n^*\right]$.
\end{lemma}
\begin{proof}
Let $v_i$ be  eigenvectors of $\lambda_i$ in \eqref{as2}, and define 
$W_n  = \mathrm{span} \{ v_i : 1 \leq i \leq m\}$.
Then, for any unit vector $x_n \in W_n $ we have 
\[
\lambda_m \leq  x_{n}^* (A\otimes I_{n}) x^*= \trace [X_n A]  \leq \lambda_1
\] 
where the both bounds converges to $f(A)$. 
\end{proof}

If $\mathcal C_m$ holds with large enough $m$ with respect to $k$, 
we can have a sequence $x_i$ in Lemma \ref{lem:condition-subspace} with orthogonal property in $\mathbb C^n$,
which is useful in proving Theorem \ref{thm:LA2}:
\begin{lemma}\label{lem:orthogonal}
Given two subspaces $W \subset \mathbb C^k \otimes \mathbb C^n$ and $T \subset \mathbb C^n$ such that $\dim W > k \dim T$, 
there exists $x \in W$ having the following Schmidt decomposition:
\begin{equation}
x = \sum_{i=1}^r \sqrt{\lambda_i} e_i \otimes f_i,
\end{equation}
Here, $\{e_i\}$ and $\{f_i\}$ are orthonormal in $\mathbb C^k$ and $\mathbb C^n$ respectively and 
moreover $f_i \perp T$, for all $i=1, \ldots, r$.
\end{lemma} 
\begin{proof}
Define $\tilde T = \mathbb C^k \otimes T$. Since $\dim \tilde T = k \dim T$, there exists a unit vector $ x \in W$ such that $x \perp \tilde T$. 
Consider the Schmidt decomposition of $x$, as in the statement, with $\lambda_i >0$, for all $i$. For any $f \in T$, we have
\begin{equation}
\langle f,  f_i  \rangle = \lambda_i^{-1/2} \langle e_i \otimes f , x  \rangle = 0,
\end{equation}
since $x \perp e_i \otimes f \in \tilde T$.
\end{proof}  

If the condition $\mathcal C_1$ is satisfied we define  a compact convex set:
\be\label{def:K}
K=\left\{B\in D_k: \trace [BA] \leq f(A) \quad\forall A\in D_k\right\}
\ee
In the following sections, we prove that both of images of mixed input states and pure input states converge to this convex set $K$. 
Especially, the latter statement is interesting because the set of pure input states itself is not a convex set. 
 
\subsection{Limiting image for mixed input states} 
Our first result is as follows: 
\begin{theorem}\label{thm:LA1}
If a sequence of projections $P_n \subset  M_{kn}$ satisfies the condition $\mathcal C_1$, then   
\be
K^\circ \subseteq \varliminf_{n \rightarrow \infty} K_n
\subseteq  \varlimsup_{n \rightarrow \infty} K_n \subseteq K
\ee
Here, as before, $K_n$ is the image of all the mixed states by the quantum channels defined by $P_n$. 
\end{theorem}

\begin{proof}
Firstly, we show that $ \varlimsup_{n \rightarrow \infty} K_n  \subseteq K$
by showing $ \varlimsup_{n \rightarrow \infty} L_n  \subseteq K$ 
because $K_n = {\rm hull} (L_n)$. 
Fix $X \in \varlimsup_{n \rightarrow \infty} L_n $ and there is a 
subsequence $\{n_j\}_j$ such that $X \in L_{n_j}$. 
Since $X$ is an output of the channel $P_{n_j}$, there exists the unit vector $x_{n_j} $ lives in the support of $P_{n_j}$ such 
that $\trace_{\mathbb C^{n_j}} [x_{n_j} x_{n_j}^*] = X$. 
By Lemma \ref{lem:upperbound}, we have $\trace \left[X A \right] \leq f(A)$ proving the result.

Secondly, we prove that $K^\circ \subseteq \varliminf_{n \rightarrow \infty}  K_n$.  
Take $X \in K^\circ$. 
Since $K$ is a compact and convex set embedded into $\mathbb R^{k^2-1}$, 
writing $r=2k^2-2$,
by Theorem \ref{Steinitz}, there exist $r$ extreme points of $K$, say, $(E_1,\ldots,E_{r})$ such that 
\be 
X \in \left({\rm hull} \{E_1,\ldots,E_{r}\}\right)^\circ
\ee 
Also, by Theorem \ref{Straszewicz}, there exists $r$-tuple of exposed points of $K$, say, 
$(F_{1},\ldots,F_r )$ such that 
\be \label{eq:hull-F}
X \in \left({\rm hull} \{F_1,\ldots,F_{r}\}\right)^\circ
\ee  
Note that since each $F_{i}$ is an exposed point of $K$,
there exists $A_{i}$ such that 
\be
\trace \left[ F_{i} A_{i}\right] &=& f\left( A_{i}\right)\notag\\
\trace \left[Y A_{i}\right] &<& f\left( A_{i}\right) \qquad \forall Y \in K \setminus \left\{F_{i}\right\} 
\label{eq:exposed}
\ee

On the other hand, by Lemma \ref{lem:condition-subspace},
for each $A_{i}$, there exists a sequence $X_{i}^{(n)} \in K_n$  such that  
\be\label{eq:lim-A}
\trace \left[ X_{i}^{(n)} A_{i} \right] \rightarrow f\left(A_{i}\right)
\qquad \text{as}\quad n\rightarrow \infty
\ee
We claim that $X_{i}^{(n)} \rightarrow F_{i}$ as $n \rightarrow \infty$.
Take a converging subsequence $X_i^{(n_j)} \to G$.
Then, the first statement: $\varlimsup_{n \rightarrow \infty} K_n \subseteq K$ implies $G \in K$ because $K$ is closed. 
Moreover, \eqref{eq:lim-A} implies that $\trace [G A_i] = f(A_i)$. 
Hence, the equation \eqref{eq:exposed} implies the above claim.

Therefore,  for large enough $n$, we have
\be
X \in {\rm hull} \{X_i^{(n)} :1\leq  i \leq r \} \subset K_{n} 
\ee
Here, the first inclusion follows from \eqref{eq:hull-F} and the second holds because $K_n$ is convex.  
Therefore, $K^\circ \subseteq \varliminf_{n \rightarrow \infty} K_n$.
\end{proof}

\subsection{Limiting image for pure input states} 
The second theorem we prove is about the image $L_n$ of the set of pure states. 

\begin{theorem}\label{thm:LA2}
If a sequence of projections $P_n \subset  M_{kn}$ satisfies the condition $\mathcal C_{m}$, with $m=(2k^2-3)k^2 + 1$, then   
\be
K^\circ \subseteq \varliminf_{n \rightarrow \infty} L_n
\subseteq  \varlimsup_{n \rightarrow \infty} L_n \subseteq K.
\ee
Here, as before, $L_n$ is the image of all the pure states by the quantum channels defined by $P_n$. 
\end{theorem}
\begin{proof}
Since condition $\mathcal C_m$ is stronger than $\mathcal C_1$, the second inclusion follows from the proof of Theorem \ref{thm:LA1}. We shall now show that the first inclusion holds.

Comparing this statement to the one in the proof of Theorem \ref{thm:LA1}, we see  that the difficulty comes from the fact that 
$L_n$ is not always a convex set. As before, for a fixed $X \in K^\circ$, choose a set of $r$ exposed points $F_1, \ldots, F_r$ of 
$K$, with $r \leq 2k^2-2$ such that $X \in \left({\rm hull} \{F_1,\ldots,F_{r}\}\right)^\circ$. The main idea here is to build 
approximating sequences $L_n \ni X_i^{(n)} \to F_i$ with an additional orthogonality property with respect to $\mathbb C^n$. 
More precisely, we want sequences $x_i^{(n)} \in \mathbb C^k \otimes \mathbb C^n$, such that
\begin{equation}
X_i^{(n)} = \mathrm{Tr}_{\mathbb C^n} \left[x_i^{(n)}{x_i^{(n)}}^*\right] \to F_i,
\end{equation}
and their Schmidt decompositions: 
\begin{equation}\label{eq:SVD-Xin}
x_i^{(n)} = \sum_a \sqrt{\lambda_a^{(i,n)}} e_a^{(i,n)} \otimes f_a^{(i,n)}
\end{equation}
have the additional property that the families $\{f_a^{(i,n)}\}_{a,i}$ are all orthogonal to each other for large enough $n$. 

Taking advantage of this orthogonality condition, we claim that, for $n$ large enough,
\begin{equation}
\mathrm{hull}\{X_i^{(n)}  : 1 \leq i \leq r\} \subset L_n.
\end{equation}
Indeed, for $X = \sum_{i=1}^r t_i X_i^{(n)}$  in the hull, the orthogonality condition implies that the unit vector
\begin{equation}
x = \sum_{i=1}^r \sqrt{t_i} x_i^{(n)}  \in \mathrm{range} P_n
\end{equation}
turns out to give $\mathrm{Tr}_{\mathbb C^n} \left[xx^*\right] = X$.
Then, as before, for large enough $n$,
\be
X \in {\rm hull} \{X_i^{(n)} :1\leq  i \leq r \} \subset L_{n} 
\ee
proving $L^\circ \subseteq \varliminf_{m \to \infty} L_m$.

To finish the proof, 
we shall construct the approximating sequences \eqref{eq:SVD-Xin}, inductively, for $i=1,2,\ldots r$. The first step is identical to the one in 
the proof of Theorem \ref{thm:LA1}: since $P_n$ satisfies $\mathcal C_1$; choose a sequence $x_1^{(n)}$ such that 
$X_1^{(n)} = \mathrm{Tr}_{\mathbb C^n} \left[x_1^{(n)}{x_1^{(n)}}^*\right]$ satisfies $X_1^{(n)} \to F_1$.
Suppose now we have constructed the first $s$ approximating sequences $x_i^{(n)}$, $1 \leq i \leq s$. For each $n$, 
Lemma \ref{lem:condition-subspace},  provides us with an $m$-dimensional subspace 
$W_n \subset \mathbb C^k \otimes \mathbb C^n$ of vectors verifying equation \eqref{eq:condition-subspace} for 
$A=A_{s+1}$. As before, one can show that for all sequences $x_n \in W_n$, the reduced states 
$X_n = \mathrm{Tr}_{\mathbb C^n} \left[x_nx_n^*\right] $ converge to $F_{s+1}$.
Define now $T_s^{(n)} = \mathrm{span}_{a, i} \{f_a^{(i,n)}\}$ with $1 \leq i \leq s$, the span of all vectors $f \in \mathbb C^n$ 
appearing in the Schmidt decompositions of the vectors $x_1^{(n)}, \ldots, x_s^{(n)}$ (see equation \eqref{eq:SVD-Xin}). 
Since $\dim T_s^{(n)} \leq sk$ and $m = (2k^2-2-1)k^2+1 > sk^2$, by Lemma \ref{lem:orthogonal}, one can find a sequence of 
vectors $x_{s+1}^{(n)} \in W_n$ such that the vectors $f$ appearing in the Schmidt decompositions are orthogonal to $T_s^{(n)}$. 

To summarize, we have constructed a sequence $x_{s+1}^{(n)}$ with the following two properties:
\begin{itemize}
\item The reduced states $X_{s+1}^{(n)}$ converge to $F_{s+1}$;
\item The vectors $f_i$ appearing in the SVD of $x_{s+1}^{(n)}$ are all orthogonal to $T_s^{(n)}$.
\end{itemize}
In such a way, one constructs recursively a family of approximating vectors with the required orthogonality condition.

\end{proof} 

\section{Asymptotic behaviour of some entropic quantities}
\label{sec:asymptotic}
\subsection{The \texorpdfstring{$S_1\to S_\infty$}{S-1 to S-infinity} norm}

Our first result is as follows

\begin{proposition}\label{infinity-norm}
Let $(\Phi_n)_n$ be a sequence of channels satisfying condition $\mathcal C_1$. Then, 
\begin{equation}
||\Phi_n||_{1,\infty}\to \max_{a \in \mathbb C^k, \, \|a\|_2=1} f(aa^*) \qquad
\text{as} \quad n \to \infty.
\end{equation}
\end{proposition}
\begin{proof}
Since it is easy to see from the definition in \eqref{def:adjoint} that
\be
\|\Phi_n\|_{1,\infty} = \|\Phi^*_n\|_{1,\infty} 
\ee
The fact that $f(\cdot)$ is convex and Proposition \ref{prop:ad} complete the proof. 
%
\end{proof}

\subsection{The minimum output entropy and the Holevo quantity}
\label{sec:S-min-chi}
 
The following proposition is a rather direct observation. 
\begin{proposition}
Let $(P_n)_n$ be a sequence of orthogonal projections satisfying condition $\mathcal C_1$. Then, one has
\begin{align}
&\lim_{n \to \infty} S_p^{\min} (\Phi_n) = S_p^{\min}(K),\\
&\lim_{n \to \infty} \chi(\Phi_n) = \chi(K).
\end{align}
\end{proposition}
\begin{proof}
By theorem \ref{thm:LA1} and continuity of the von Neumann entropy, the first statement is proved. 
For the second, note that Carath\'eodory's theorem implies that 
optimal ensemble can always consist of $2k^2+1 $.
Therefore, the first statement also implies the second. 
\end{proof}

The following proposition gives a necessarily and sufficient condition for the Holevo capacity to be written nicely: 
\begin{proposition}\label{prop:holevo-smin}
We have
\begin{equation} \label{formula:holevo-smin}
\chi(K) = \log k - S^{\min}(K).
\end{equation}
if and only if the limiting convex set $K$ has the property that
\begin{equation}
I/k \in \mathrm{hull} \, (\mathrm{argmin} \, S),
\end{equation}
where $X \in \mathrm{argmin} \, S$ if and only if $S(X) = S^{\min}(K)$,
\end{proposition}
\begin{proof}
The first and second terms in \eqref{HC2} have the upper bounds respectively
$\log k$ and $-S^{\min} (K)$ for the convex set $K$. 
So, our assumption let $\Phi$ achieve the both bounds. 
The converse is obvious from this argument. 
\end{proof} 

Suppose $f$ is $G$-invariant for some group $G$ and its unitary representation $\{U_g\}_{g\in G}$,
i.e.,  $f(U_g A U_g^*) = f(A)$ for all $g\in G$ and $A \in D_k$. Then $K$ is invariant with respect to those rotations: 
$U_g K U_g^* =K$ for all $g\in G$. This, in particular, implies that the set of optimal points $\mathrm{argmin} \, S$ is also invariant. 
In addition, if the unitary representation $\{U_g\}_{g \in G}$ is irreducible so that $\int U_g A U_g^* = I/k$ for all $A\in D_k$ 
\cite{holevo}, then we get the formula \eqref{formula:holevo-smin}. For example, consider the additive group 
$\mathbb Z_k \times \mathbb Z_k$ and define unitary operators, which are called \emph{discrete Weyl operators}, by 
\be W_{a,b} = X^a Y^b \label{formula:weyl}\ee 
Here, $(a,b) \in \mathbb Z_k \times \mathbb Z_k$, and $X$ and $Y$ act on the canonical basis vectors $\{e_l\}_{l=1}^k$ of 
$\mathbb C^k$ as follows:
\be
X e_l = e_{l+1} \quad \text{ and } \quad Y e_l = \exp\left\{\frac{2\pi \mathrm i}{k} \cdot l \right\}
\ee
This is an irreducible unitary adjoint representation of the group $\mathbb Z_k \times \mathbb Z_k$ on $\mathbb C^k$.  
Although this argument only gives a sufficient condition, it turns out to be useful. 
We have proven the following corollary. 
\begin{corollary} \label{cor:holevo-smin}
Suppose, as is in \eqref{def:K}, a convex set $K$ is defined by a function $f$ 
which is invariant with respect to the discrete Weyl operators:
\be\label{invariance-f}
f(W_{a,b}\, A \, W_{a,b}^*) = f(A)  \qquad \forall A \in D_k, \quad \forall (a,b) \in \mathbb Z_k \times \mathbb Z_k .
\ee
Then, the formula \eqref{formula:holevo-smin} holds.
\end{corollary} 

\section{Free probability}\label{sec:free-prob} 

A \emph{$*$-non-commutative probability space} is
a unital $*$-algebra $\mathcal A$ endowed with a linear map $\phi\colon\mathcal A\to\mathbb C$
satisfying $\phi (ab)=\phi (ba),\phi (aa^{*})\geq 0, \phi (1)=1$.
The map $\phi$ is called a trace, and
an element of $\mathcal A$ is called
a \emph{non-commutative random variable}. 

Let $\mathcal A_1, \ldots ,\mathcal A_k$ be subalgebras of $\mathcal A$ having the same unit as $\mathcal A$.
They are said to be \emph{free} if for all $a_i\in \mathcal  A_{j_i}$ ($i=1, \ldots, k$) 
such that $\phi(a_i)=0$, one has  
$$\phi(a_1\cdots a_k)=0$$
as soon as $j_1\neq j_2$, $j_2\neq j_3,\ldots ,j_{k-1}\neq j_k$.
Collections $S_{1},S_{2},\ldots $ of random variables are said to be 
$*$-free if the unital $*$-subalgebras they generate are free.

Let $(a_1,\ldots ,a_k)$ be a $k$-tuple of self-adjoint random variables and let
$\mathbb{C}\langle X_1 , \ldots , X_k \rangle$ be the
free $*$-algebra of non commutative polynomials on $\mathbb{C}$ generated by
the $k$ self-adjoint indeterminates $X_1, \ldots ,X_k$. 
The {\it joint distribution\it} of the $k$-tuple $(a_i)_{i=1}^k$ is the linear form
\begin{align*}
\mu_{(a_1,\ldots ,a_k)} : \C\langle X_1, \ldots ,X_k \rangle &\to \C \\
P &\mapsto \phi (P(a_1,\ldots ,a_k)).
\end{align*}

Given a $k$-tuple $(a_1,\ldots ,a_k)$ of free 
random variables such that the distribution of $a_i$ is $\mu_{a_i}$, the joint distribution
$\mu_{(a_1,\ldots ,a_k)}$ is uniquely determined by the
$\mu_{a_i}$'s.

Considering a sequence of $k$-tuples $(a_i^{(n)})_{i=1}^k$ in $*$-non-commutative probability spaces
$(\mathcal{A}_n,\phi_n)$, we say that it converges \emph{in distribution} to 
the distribution of  $(a_1,\ldots ,a_k)\in (\mathcal{A},\phi)$ iff $\mu_{(a_1^{(n)},\ldots ,a_k^{(n)})}$
converges point wise to $\mu_{(a_1,\ldots ,a_k)}$.
Likewise, a sequence is said to converge \emph{strongly in distribution} iff it converges in distribution, and 
in addition, for any non-commuytative polynomial $P$, its operator norm converges
$$\|P(a_1^{(n)},\ldots ,a_k^{(n)})\|\to \|P(a_1,\ldots ,a_k)\|.$$
In this definition, we assume that the operator norm is given by the distribution, i.e.
$$\|P(a_1^{(n)},\ldots ,a_k^{(n)})\|=\lim_p\|P(a_1^{(n)},\ldots ,a_k^{(n)})\|_p,$$ 
and 
\be\label{norm-limit}
\|P(a_1,\ldots ,a_k)\|=\lim_p\|P(a_1,\ldots ,a_k)\|_p
\ee

For the purpose of this paper, let us record  two important theorems which extend strong convergence.
I.e.,
let  $(a_i^{(n)})_{i=1}^k$ be a sequence of $n\times n$ matrices, viewed as elements of the non-commutative
probability space $(M_n,n^{-1}Tr)$ and assume that it converges strongly in distribution towards
a $k$-tuple of random variables
$(a_1,\ldots ,a_k)\in (\mathcal{A},\phi)$,
then we have the following extension theorems. 

\begin{theorem}\label{thm:cm}
Let $U_n$ be an $n\times n$ Haar distributed unitary random matrix. Then the family
$$(a_1^{(n)},\ldots ,a_k^{(n)},U_n,U_n^*)$$
almost surely converges strongly too, towards the $k+2$-tuple of random variables
$(a_1,\ldots ,a_k,u,u^*)$, where $u,u^*$ are unitary elements free from $(a_1,\ldots ,a_k)$
\end{theorem}
Historically, the convergence of distribution is due to Voiculescu, \cite{voiculescu98}. 
A simpler proof was given by \cite{collins-imrn}.
The strong convergence relies on \cite{cm} - it relies heavily on preliminary works by \cite{HT} and \cite{male}.

Actually,  although this is counterintuitive, Theorem \ref{thm:cm} is equivalent to a stronger statement where, 
in the conclusion, the non-commutative polynomial is not taken with complex coefficients, but with any
matrix coefficient of fixed size. This follows from the ``Linearization Lemma'' as proved by  Haagerup and Thorbj\o{}rnsen \cite{HT}.
We state this result below, as it will be useful to widen our range of examples. 

\begin{theorem}\label{thm:strong-matrix-coeff}
Let $P$ be a non-commutative polynomial in $k$ variables with coefficients in $M_l (\mathbb C)$ instead of $\mathbb C$.
Then the operator norm of $P((a_i^{(n)}))_{i=1}^k\in M_l\otimes M_n$ still converges as $n\to\infty$. 
The limit is obtained by taking the limit as $p\to \infty$ of the limit as $n\to\infty$ of the $p$-norms.
\end{theorem}

\section{Example of non-random projections} 
\label{sec:non-random} 
In this section we consider some elementary examples of deterministic sequence of projections which satisfy the condition $C_m$. 

Let's start with the completely depolarizing channel $\Phi_n: M_N \to M_k$:
\be\label{eq:dep}
\Phi(\rho) = \trace [\rho] \cdot I_k /k
\ee 
Its adjoint channel is written as
\be
\Phi^* (\sigma) = \trace [\sigma] \cdot I_N/k
\ee
This immediately implies via Proposition \ref{prop:ad} the following result:
\begin{proposition}
The depolarizing channels defined in \eqref{eq:dep} satisfies the condition $\mathcal C_m$
for all $m\geq 1$. 
\end{proposition}
Note that the above example is trivial, since the image of the channels $\Phi_n$ consists of a single point, 
$\{I_k/k\}$, so the convergence result is obvious.

We now generalize the above example by considering a subclass of entanglement-breaking channels. 
In general, any entanglement-breaking channel has the Holevo form \cite{HSR}:  
\be\label{holevo-form}
\Xi (X) = \sum_{i=1}^l  \trace [X M_i] \sigma_i
\ee
where $\{M_i\}_i$ are positive operators which sum up to the identity and $\sigma_i$ are fixed states. 
Note that the set of the operators $\{M_i\}_i$ is called \emph{Positive Operator Valued Measure} in quantum information theory
and $p_i = \trace [X M_i]$ constitute a probability distribution. 
As the name suggests, those channels break entanglement through measurements. 

\begin{proposition}
Let $\Phi_n$ be a sequence of entanglement-breaking channels:
\be
\Phi_n (\rho) = \sum_{i=1}^l \trace [M_i^{(n)} \rho] \sigma_i  
\ee
where $l>0$ and $(\sigma_i)_{i=1}^l \in D_k^{l}$ do not depend on $n$,
and, for all $n \in \N$, $(M_i^{(n)})_{i=1}^l$ is a  POVM such that $\|M_i^{(n)}\|=1$ for all $i=1,\ldots,l$.
Then, the sequence of projections $P_n$ associated to $\Phi_n$ satisfies the condition $\mathcal C_m$ where
\be
m = \liminf_{n \to \infty} \min_{1\leq i \leq l} \dim_1 M_i^{(n)} \geq 1.
\ee 
Here, for a given operator $X$, $\dim_1 X$ denotes the dimension of the eigenspace corresponding to the eigenvalue $\lambda=1$ of $X$. 
\end{proposition}
\begin{proof}
A direct computation shows that the adjoint channel of $\Phi_n$ is 
\begin{equation}
\Phi_n^*(A) = \sum_{i=1}^l \mathrm{Tr}[A \sigma_i] M_i^{(n)}.
\end{equation}
First, note that the operator $\Phi_n^*(A)$ has eigenvalue $\mathrm{Tr}[A \sigma_i]$ with multiplicity $ \dim_1 M_i^{(n)}$. 
Define now
\begin{equation}
f(A) = \max_{1 \leq i \leq l} \mathrm{Tr}[A \sigma_i].
\end{equation}
It follows that $\Phi_n^*(A)$ has eigenvalue $f(A)$ with multiplicity at least
\begin{equation}
m_n = \min_{1\leq i \leq l} \dim_1 M_i^{(n)} \geq 1.
\end{equation}
Also, we claim that $\|\Phi_n^*(A)\| = f(A)$:
\begin{equation}
\Phi_n^*(A) = \sum_{i=1}^l \mathrm{Tr}[A \sigma_i] M_i^{(n)} \leq \sum_{i=1}^l f(A)M_i^{(n)} = f(A)I_N.
\end{equation}
We conclude that $f(A)$ is the largest eigenvalue of $\Phi^*(A)$ and that it has multiplicity at least $m_n$; the conclusion 
follows by Proposition \ref{prop:ad}. 
\end{proof}
\begin{remark}
The condition $\|M_i^{(n)} \|=1$ ensures that the image of the channel $\Phi_n$ is precisely $\mathrm{hull}(\sigma_i)_{i=1}^l$, 
and thus the convergence to the limiting set $K$ is again obvious. 
\end{remark}

\section{Examples of random projections}
\label{sec:examples}
In this section we look at random projection operators and we show how Theorem \ref{thm:LA1} together with 
Theorems \ref{thm:cm} and  \ref{thm:strong-matrix-coeff} give interesting examples. 

\subsection{Random Stinespring channels}
\label{sec:examples-Stinespring}
Let us first study  channels coming from random isometries. Such random channels were  used by 
Hayden and Winter \cite{hayden-winter} to show violations of additivity for minimum $p$-R\'enyi entropy, for $p$ close to $1$. 
Following Hastings' counterexample (see the next subsection), it was shown that they also violate additivity for the von Neumann 
entropy ($p=1$) \cite{fki,bho,asw}. More recently, the output of these channels has been fully characterized using free 
probability theory \cite{bcn} and macroscopic violations (or order of $1$ bit) for the additivity of the MOE have been observed \cite{bcn13}.

We construct the channel from the Stinespring dilation
\begin{equation}\label{eq:Stinespring-channels}
\Phi_n(X) = [\mathrm{id} \otimes \mathrm{Tr}](VXV^*),
\end{equation}
where $V: \mathbb C^N \to \mathbb C^k \otimes \mathbb C^n$ is a random Haar isometry. In particular, the operator 
$P_n= V_n V_n^*$ projects onto a random Haar $N$-dimensional subspace of $\mathbb C^k \otimes \mathbb C^n$. 
The asymptotic regime is as follows: we fix a parameter $t\in (0,1)$, and $N$ is any function of $n$ that satisfies
$N\sim tnk$.

Under these circumstances, the convex set $K$ defined in \eqref{def:K} is renamed $K_{k,t}$ and it was studied at 
length in \cite{bcn} and \cite{bcn13}.

\begin{proposition}\label{prop:f-Stinespring-channels}
Consider the free product $\mathcal M$ of the von Neumann non-commutative probability spaces 
$(M_k(\mathbb C), \mathrm{tr})$ and $(\mathbb C^2, t\delta_1+(1-t)\delta_2)$. The element $p=(1,0)$ of $\mathbb C^2$ in 
$\mathcal M$ is a selfadjoint projection of trace $t$, free from elements in $M_k(\mathbb C)$. 
For any $A\in M_k(\mathbb C)$, we define  $f_t(A) = \|pAp\|$. Then, the sequence of projections $P_n$ 
defining the quantum channels \eqref{eq:Stinespring-channels}
satisfies condition $\mathcal C_m$ for any $m$, with limiting function $f_t$.
\end{proposition}

A proof can be deduced from the next section on mixed unitary channels. 
We also refer the  reader to \cite{bcn,bcn13} for the proof of the following theorem, gathering some of the most important 
properties of the set $K_{k,t}$. As an original motivation, let us state the following theorem, in which the element with 
least entropy inside $K_{k,t}$ is identified.

\begin{theorem}
The convex set $K_{k,t}$ has the following properties:
\begin{enumerate}
\item It is conjugation invariant: $A \in K_{k,t} \iff UAU^* \in K_{k,t}$, for all $U \in \mathcal U(k)$. In particular, one only 
needs the eigenvalues of a selfadjoint element in order to decide if it belongs to $K_{k,t}$ or not.
\item Its boundary is smooth iff $t<k^{-1}$.
\item Any self-adjoint element with eigenvalues
$$\lambda = (a, \underbrace{b, b, \ldots, b}_{k-1 \text{ times}}),$$
where 
$$a = \begin{cases}
t + \frac 1 k -\frac{2t}{k} +2 \frac{\sqrt{t(1-t)(k-1)}}{k},&\qquad \text{if } t+\frac{1}{k} < 1\\
1,&\qquad \text{if } t+\frac{1}{k} \geq 1\\
\end{cases}$$
and $b=(1-a)/(k-1)$ is a \emph{joint} minimizer for all the $p$-R\'enyi entropies on $K_{k,t}$, for all $p \geq 1$.
\end{enumerate}
\end{theorem}

\subsection{Random Mixed Unitary Channels} 
\label{sec:examples-RMUC}
In this section, we are interested in  \emph{random mixed unitary channels}, namely, convex combinations of random 
automorphisms of $M_n(\mathbb C)$ (note that that deterministic incarnations of these these channels are also 
known in the literature as ``random unitary channels''; in this work, we prefer the term ``mixed'', since the unitary operators 
appearing in the channel are themselves random). 
After the setup, we argue that this class of channel has the property $\mathcal C_m$ for all $m$. 
Based on this result, we identify the limiting maximum output infinity norm of this class. 
This section ends with the assertion that this class satisfies the property \eqref{formula:holevo-smin},
which gives a simple relation between MOE and HC. 

To set up our model, we recall that these channels can be written as follows
\begin{eqnarray*}\tilde \Phi_{n,k}^{(w)}:M_n(\mathbb C)\to M_n(\mathbb C)\\
\tilde \Phi_{n,k}^{(w)} (X)= \sum_{i=1}^k w_i U_iXU_i^*,
\end{eqnarray*}
and we are interested in the complementary channels; $\tilde {\tilde \Phi} = \Phi$.
Here, $\{U_i\}_{i=1}^k$ are i.i.d.~ $n\times n$ Haar distributed random unitary matrices and $w_i$ are positive weights 
which sum up to one (we shall consider the probability vector $w$ a parameter of the model). Here, $N=n$ and the 
corresponding isometric embedding is the block column matrix whose $i$-th block is $\sqrt{w_i}U_i$.
Then, the corresponding projection $P_n^{(w)}\in M_n(\mathbb C)\otimes M_k(\mathbb C)$ is given by 
\begin{equation}
P_n^{(w)} = \sum_{i,j=1}^k\sqrt{w_iw_j}\,e_ie_j^* \otimes  U_iU_j^* ,
\end{equation}
where $\{e_i\}$ is the canonical basis of $\mathbb C^k$.
Our model of this paper corresponds to the complementary channel of this channel
$ \Phi_{n,k}^{(w)}:M_n(\mathbb C)\to M_k(\mathbb C)$,
such that the matrix entries of its output are as follows
\be \label{mixed-U}
\left( \Phi_{n,k}^{(w)} (X)\right)_{ij}=\sqrt{w_iw_j}\trace  \left[U_iXU_j^*]\right.
\ee

Firstly, we claim that 
these sequences of channels almost surely have property $\mathcal C_m$ with $m \geq 1$,
and moreover, the limiting function can be written explicitly as follows. 
Let $L(F_k)$ be the free group von Neumann algebra with $k$ free generators $u_1,\ldots ,u_k$. Consider the
algebra $M_k (L(F_k))$. This algebra contains $M_k(\mathbb C)$ in a natural way, and for $A\in D_k$ with respect to this inclusion, define
\begin{equation}\label{eq:def-f-w}
f_w(A)=\|P^{(w)}AP^{(w)}\|,
\end{equation}
where, for all $i,j$,
$P_{ij}^{(w)}=\sqrt{w_iw_j}u_iu_j^*$.
Then, our first claim is:
\begin{proposition}\label{prop:unitary-mix}
 The sequence of orthogonal projections $P_n^{(w)}$ almost  surely satisfies condition 
$\mathcal C_m$ for all $m$ 
 with limiting function $f_w$ defined in \eqref{eq:def-f-w}.
\end{proposition} 
\begin{proof}
First, notice that $P_n^{(w)}(A\otimes I_n ) P_n^{(w)}$ can be understood as polynomials of $P_n$'s 
with coefficients in $M_k$. Indeed, 
\be
P_n^{(w)}(A\otimes I_n ) P_n^{(w)} 
= \sum_{i,j,s,t=1}^k\sqrt{w_iw_jw_s w_t} \left( e_ie_j^* \, A \, e_se_t^* \right) \otimes  U_iU_j^*U_sU_t^*  
\ee
Hence, Theorem \ref{thm:strong-matrix-coeff} implies that, for any fixed matrix $A\in M_k$,
the operator norm of $P_n^{(w)}(A\otimes I_n ) P_n^{(w)}$ converges to 
the operator norm of $P^{(w)}AP^{(w)}$ because it 
follows from Theorem \ref{thm:cm} that
$k$ independent random unitary matrices $(U^{(n)}_i)_{i=1}^k$  strongly converges to a $k$-tuple of free unitary elements $(u_i)_{i=1}^k$ almost surely.

We have thus shown that, for every matrix $A \in D_k$, almost surely, $\|P_n^{(w)}(A\otimes I_n ) P_n^{(w)}\| \to f_w(A)$. To conclude that the property $\mathcal C_1$ holds, we have to show the above convergence simultaneously, for all $A$. To do this, consider a countable set $(A_i) \subset D_k$ with the property that for all $A \in D_k$ and for all $\varepsilon >0$, there is some $i$ such that $\|A-A_i \| < \varepsilon$. By taking a countable intersection of probability one events, the convergence $\|P_n^{(w)}(A_i\otimes I_n ) P_n^{(w)}\| \to f_w(A_i)$ holds almost surely, for all $i \geq 1$. 
Since the function $f(\cdot)$ is continuous from the definition,  
the above chosen sequences of projections show the convergence
$\|P_n^{(w)}(A\otimes I_n ) P_n^{(w)}\| \to f_w(A)$ for all matrices $A \in D_k$.

Next, we show $\mathcal C_m$ property with $m>1$. 
Remember that the infinity norm is defined by the limit of $p$-norms as in \eqref{norm-limit},
the limiting density function yields non-vanishing measure around the limiting infinity norm. 
More precisely, for any $\epsilon>0$ there exists a ratio $0<\eta_\epsilon<1$ such that 
the measure of the $\epsilon$-neighborhood of the infinity norm is $\eta_\epsilon$. 
Hence, fix $A\in D_k$ and for large enough $n$, we have $\eta_\epsilon \cdot n$ eigenvalues of $P_n^{(w)}(A\otimes I_n ) P_n^{(w)}$ 
which are $2 \epsilon$-close to the limiting infinity norm. 
As $\epsilon > 0$ is arbitrary, for any $m \geq 1$, the largest $m$ eigenvalues converge to the operator norm almost surely. 
Again, we can prove that almost surely all the sequences show this convergence for all $A\in M_k$. 
This proves  $\mathcal C_m$ property with $m>1$. 
\end{proof}

Secondly, we characterize the limiting value of the maximal output infinity norm via Proposition \ref{prop:unitary-mix}. 
To do so, we recall the following result from \cite{aos}, generalizing questions that can be traced back to \cite{kes} 
(for a matricial coefficient version, see \cite{leh}):
\begin{proposition}\cite[Theorems IV G and IV K]{aos}\label{prop:sum-of-unitaries}
Consider an integer $k \geq 2$ and let $\{u_1, \ldots, u_k\}$ be a family of free unitary random variables and 
$a = (a_1, \ldots, a_k)$ a scalar vector. Then,
\begin{equation}
\psi(a):=\left\|\sum_{i=1}^k a_i u_i \right\| = \min_{x \geq 0} \left[ 2x + \sum_{i=1}^k \left( \sqrt{x^2 + |a_i|^2} - x\right)\right].
\end{equation}
Moreover, 
\begin{align}
\min_{\|a\|_1 = 1} \psi(a) &= \frac{2 \sqrt{k-1}}{k}\\
\max_{\|a\|_2 = 1} \psi(a) &= \frac{2 \sqrt{k-1}}{\sqrt k}
\end{align}
with both extrema being achieved on ``flat'' vectors, i.e. vectors with $|a_i| = \text{const}$.
\end{proposition}

\begin{remark}\label{rem:positive-x}
In \cite{aos}, the minimum in the formula for $\psi$ is taken over all values $x \geq 0$, 
but one can show, by considering the derivative of the above function at $x=0$, 
that the minimum is achieved at a strictly positive value $x>0$ for $k \geq 3$. 
\end{remark}

Let us introduce the following notation: for a given vector $b \in \mathbb C^k$, let
\begin{equation}\label{function:psi*}
\psi_*(b) = \sup_{\|a\|_2=1} \psi(a.b),
\end{equation}
where $(a.b)_i = a_i b_i$. From the result above, we have that
\begin{equation}
\psi_*((\underbrace{1,\ldots, 1}_{k \text{ times}})) = \frac{2 \sqrt{k-1}}{\sqrt k}.
\end{equation}

Next, from this proposition, we can show the following results:
\begin{theorem}\label{thm:random-mixed-unitary-channels}
For $k \geq 3$, 
\begin{enumerate}
\item
The function $f_w$ defined in \eqref{eq:def-f-w} satisfies 
\begin{equation}
\max_{\substack{A \in D_k \\ \rk A=1}} f_w(A)  = \psi_*(\sqrt w),
\end{equation}
where $\sqrt w \in \ell^2$ is the vector 
with coordinates $(\sqrt{w})_j = \sqrt{w_j}$. The general formula for $\psi_*(\sqrt w)$ is described in the Appendix, 
Proposition \ref{prop:optimization-g}. 
\item
This 
implies that, with probability one, 
$\|\Phi_{n,k}^{(w)}\|_{1,\infty}$ converges to $\psi_*(\sqrt w)^2$  as $n \to \infty$.
\item
In the particular case of the flat distribution $w = w_\text{flat} = (1/k, \ldots, 1/k)$, we have 
\begin{align}
\psi_*(\sqrt{w_\text{flat}}) =\psi(w_\text{flat})  &=  \frac{2 \sqrt{k-1}}{k}\\
\lim_{n \to \infty} \|\Phi_{n,k}^{(w_\text{flat})}\|_{1,\infty} &= \frac{4(k-1)}{k^2}.
\end{align}
\end{enumerate}
\end{theorem}

\begin{proof} 
Since $A \in D_k$ is a pure state, it can be written as $A=aa^*$ for some unit vector $a = (a_1,\ldots, a_k) \in \C^k$. 
Then, since we work in $C^*$-algebra, 
\begin{align}
f_w(A) &= \|P^{(w)}aa^*P^{(w)}\|_{M_k(L(F_k))} = \|P^{(w)}a\|^2_{\ell^2(L(F_k))} \\
&= \sum_{i=1}^k  \|[P^{(w)}a]_i\|^2_{L(F_k)} =  \sum_{i=1}^k  \|\sqrt{w_i}u_i \sum_{j=1}^k \sqrt{w_j}\bar{a_j}u_j^*\|^2\\
&=  \sum_{i=1}^k w_i \|\sum_{j=1}^k a_j \sqrt{w_j} u_j\|^2 = \|\sum_{j=1}^k a_j \sqrt{w_j} u_j\|^2 = \psi(a.\sqrt w)^2.
\end{align}
Taking the supremum over all $a \in \C^k$ with $\|a\|_2=1$ proves the first claim.
The second one is a consequence of  Proposition \ref{infinity-norm} and the third is shown by
Proposition \ref{prop:sum-of-unitaries}.
\end{proof}

\begin{remark} Note that the function $f_w$ is not ``spectral'', i.e. it does not depend only on the spectrum of its input, as it is the case for the function $f_t$ from Proposition \ref{prop:f-Stinespring-channels}. Indeed, notice that, with the choice of the unit vectors
\begin{align}
a^{(1)} &= (1, 0, \ldots, 0)\\
a^{(2)} &= (1/\sqrt{k}, 1/\sqrt{k}, \ldots, 1/\sqrt{k}),
\end{align}
one has 
\begin{equation}
1 = f_w \left(a^{(1)}\left(a^{(1)}\right)^*\right) \neq f_w \left(a^{(2)}\left(a^{(2)}\right)^*\right) = \frac{2 \sqrt{k-1}}{\sqrt k},
\end{equation}
although the matrices $A_i = a_ia_i^*$ are isospectral.
\end{remark}

Thirdly, we claim that, in the limit, the minimum output entropy and the Holevo capacity of  the channel \eqref{thm:strong-matrix-coeff}
identify each other:
\begin{theorem} 
The convex set $K$ for $\Phi^{(w)}_n$ has the property \eqref{formula:holevo-smin}.
\end{theorem}
\begin{proof}
Take the Weyl operators $W_{a,b}$ as defined in \eqref{formula:weyl} and calculate
\be
\left \| P_n (W_{a,b} A W_{a,b}^* \otimes I_n) P_n\right\|_\infty 
=\Big \| \left(\sqrt A \otimes I_n \right)\underbrace{ \left(W_{a,b}^* \otimes I_n \right) P_n \left(W_{a,b} \otimes I_n \right)}_{(\star)} ( \sqrt A  \otimes I_n) \Big\|_\infty 
\label{rotationally-invariant}
\ee
while we have 
\be
(\star)&=&(W_{a,b}^* \otimes I_n ) \left( \sum_{s,t=1}^k \sqrt{w_s w_t} \,e_s e_t^* \otimes U_s U_t^*  \right) (W_{a,b} \otimes I_n) \\
&=&  \sum_{s,t=1}^k \sqrt{w_s w_t} \exp \left\{ \frac{2 \pi \mathrm i }{n} \,b(t-s) \right\}  e_{s-a} e_{t-a}^* \otimes U_s U_t^* \\
&=&  \sum_{s,t=1}^k \sqrt{w_s w_t}  e_{s-a} e_{t-a}^* \otimes \left( \exp \left\{ -bs \frac{2 \pi \mathrm i }{k}  \right\} U_s \right)
\left(\exp \left\{ -bt \frac{2 \pi \mathrm i }{k}  \right\} U_t\right)^*
\ee
This implies that  $(W_{a,b}^* \otimes I_n ) P_n ( W_{a,b} \otimes I_n ) $ 
have the same law for all $(a,b) \in \Z_k \times Z_k$ 
because $\{U\}_{i=1}^k$ are i.i.d. with respect to the Haar measure. 
Therefore,  
\eqref{invariance-f} is true and then Corollary \ref{cor:holevo-smin} completes the proof.  
\end{proof}  

\section{Image of tensor product of channels} 
\label{sec:tensor}
In this section, we investigate the image of tensor products of two channels. 
Section \ref{sec:tensor1} describes general theory when one channel has a nice asymptotic behavior and the other is fixed. 
In Section \ref{sec:tensor2}, we consider cases where the fixed channel is entanglement-breaking. 
\subsection{Tensor with any finite dimensional quantum channel}
\label{sec:tensor1}
Our setting is as follows. 
Let $\Psi_n$ be a quantum channel obtained from $P_n$ a sequence of projections in $M_k\otimes M_n$ of rank $N=N(n)$ , namely,
$$\Psi_n: M_N\to M_k.$$
Then, 
\begin{theorem}\label{thm:product-image}
If the family $(P_n, E_{ij}\otimes I_n : i,j\in \{1,\ldots ,k\})$ converges strongly 
as in the definition of section \ref{sec:free-prob}
then, 
for any quantum channel $\Xi: M_p\to M_q$,
there exists a convex body $K$ in $D_{kq}$ such that
$$\Xi\otimes \Psi_n (S_{pN})\to K$$ 
as in Theorem \ref{thm:LA2}.
\end{theorem}
\begin{remark}
In this setting, existence of the limiting convex set of output states of the tensor products 
depend only on the asymptotic behavior of $\Psi_n$.
\end{remark}
   
 \begin{proof}
 First, we choose $m \in \mathbb N$  such that there exists a projection $P$ of rank $p$ on  $M_q\otimes M_m$ which is associated to $\Xi$.
 This construction can be made uniquely up to an isometry between $Im (P)$ and $\C^p$.  
 
Next,
it follows from theorem \ref{thm:strong-matrix-coeff} that the fact that
$(P_n, E_{ij}\otimes I_n : i,j\in \{1,\ldots ,k\})$ converges strongly as $n\to \infty$
implies also that
$$(P_n\otimes P, E_{i_1j_1}\otimes 1_n\otimes E_{i_2j_2} : i_1,j_1\in \{1,\ldots ,k\}, i_2,j_2\in 
\{1,\ldots qm\})$$
converges also strongly.
This strong convergence implies that for any $A\in M_k\otimes M_q$, the sequence
$P_n\otimes P A\otimes 1_{nm}P_n\otimes P$
satisfies the condition $\mathcal C_l$ (see Definition \ref{def:condition})
for any $l$. Note that in the above equation, we viewed $A\otimes 1_{nm}$
an an element of $M_k\otimes M_n\otimes M_q\otimes M_n$.

Finally, 
the proof then follows from Theorem \ref{thm:LA2}
\end{proof}

We want to point out that 
it remains difficult to analyze the limiting outputs sets $K$ of Theorem \ref{thm:product-image} in general. 
For example, even in the simple case
where $\Xi$ is the identity map,
we are unable to describe the collection of limiting output sets.

\subsection{Tensor with entanglement breaking channel}
\label{sec:tensor2}
It seems difficult in general to compute $K$ explicitly in the tensor product case. 
However, when $\Xi$ is an entanglement-breaking channel of certain type, 
we can write down the image explicitly. 
In this section, channels tensored with entanglement-breaking channels are fixed and we do not use the asymptotic behavior 
to get results, in the first place. 

Let us start with an interesting example among entanglement-breaking channels, which is called \emph{pinching map}:
\be
\Xi : m_{i,j} \mapsto \delta_{i,j} m_{i,j} 
\ee
where $m_{i,j}$ is the $(i,j)$-element of square matrices. 

\begin{proposition}
Let $\Xi: M_l\to M_l$ be the pinching map. Then the image $K_{\Xi \otimes \Psi}$ can be described as follows.
$$\tilde K=\{a_1K_\Psi \oplus \ldots \oplus a_lK_\Psi ; (a_i)\in \Delta_l\}$$
\end{proposition}
\begin{proof}
This follows directly from the fact that the image of $S_{lN}$ under $\Xi\otimes 1_N$ is exactly 
$\{a_1S_N\oplus \ldots \oplus a_lS_N,(a_i)\in \Delta_l\}$.
This can be readily seen by double inclusion.
\end{proof}
This has an immediate corollary:
\begin{corollary}
Let $\Psi_n$ be a sequence of quantum channels obeying the hypotheses of Theorem \ref{thm:LA1} (or Theorem \ref{thm:LA2}).
Then, for any integer $l$, the conclusion of Theorem \ref{thm:LA1} (or Theorem \ref{thm:LA2}) still holds true for $\Psi_n^{\oplus l}$,
where $K$ is replaced by $K^{\oplus l}$.
\end{corollary}

The images of entanglement-breaking channels are described as follows:
\begin{lemma}\label{image-EB}
For an entanglement-breaking channel $\Xi$ defined in \eqref{holevo-form}.
Then, $K_\Xi =\Xi (S_{p})$ is written as 
\be
K_\Xi = \left\{ \sum_{i=1}^l p_i \sigma_i : \, (p_i) \in \Delta_{\Xi}\right \}  
\ee 
Here, we denote possible  probability distributions by channel $\Xi$ by $\Delta_\Xi$. 
\end{lemma}

A straightforward application of Lemma \ref{image-EB} implies:
\begin{lemma}
Suppose we have two quantum channels $\Xi$ and $\Psi$. 
Let $\Xi$ be an entanglement-breaking  channel defined in \eqref{holevo-form}.
Then, the set of images of  all the states via $\Xi \otimes\Psi$ is given by
\be\label{image-conj1}
 K_{\Xi \otimes \Psi} ={\rm hull} \left\{ \sum_{i=1}^l \sigma_i \otimes \Psi \left (BM_i^T B^*\right)  :\, \text{$B \in M_{N, p}$ with $\trace [BB^*]=1$.} \right\}
\ee
\end{lemma}
\begin{proof}
Let the input spaces of $\Xi$ and $\Psi$ be $\C^p$ and $C^N$, respectively. 
Take a bipartite vector $b$ in $\C^N \otimes \C^p$ and calculate as follows.
\be
(\Xi \otimes \Psi) (bb^*) &=&\sum_{i=1}^l \sigma_i \otimes \Psi \left(BM_i^TB^*\right) 
\ee
where we used the canonical isomorphism: $\C^N \otimes \C^p \ni b \leftrightarrow B \in  M_{N, p}(\C) $.
Indeed, 
\be
\trace_{\C^l} \left[ bb^* (M_i \otimes I_N) \right] = BM_i^TB^*
\ee
\end{proof} 

Then, we define 
\be\label{image-conj2}
K_\Xi \otimes K_\Psi
\ee
such that 
$\otimes$ in the formula yields the smallest convex set which contains all the simple tensors. 
It is easy to see that $\eqref{image-conj2} \subseteq \eqref{image-conj1}$:
\be
K_\Xi \otimes K_\Psi \subseteq  K_{\Xi \otimes \Psi} 
\ee
These two sets turn out to be identical under some assumption:
\begin{theorem}\label{thm:EB1}
Suppose we have two quantum channels $\Xi$ and $\Psi$. 
Let $\Xi$ be an entanglement-breaking  channel defined in \eqref{holevo-form} such that  $\|M_i\|_\infty =1$ for $1 \leq i \leq l$.
\be
K_{\Xi \otimes \Psi} = K_\Xi \otimes K 
\ee
\end{theorem} 
\begin{proof}
 We show $K_\Xi \otimes K_\Psi \supseteq  K_{\Xi \otimes \Psi} $.  
This is true if for all $B$
there exist $\{r_k\} \in \Delta_d$, $\{p^{(k)}_i\}_i \in \Delta_\Xi$, $\rho^{(k)} \in S$ such that 
\be\label{EB-1}
\sum_{k=1}^d r_k \left( \sum_{i=1}^l p^{(k)}_i \sigma_i \otimes \Psi\left( \rho^{(k)}\right)   \right) =\sum_{i=1}^l \sigma_i \otimes \Psi \left(BM_i^TB^* \right)
\ee
and this is true if 
\be
\sum_k  r_k p^{(k)}_i \rho^{(k)} = BM_i^TB^*  \qquad \forall i \in \{1,\ldots,l\}
\ee
This can be written, by abusing notations, as 
\be\label{formula:matrixeq} 
P\cdot \Gamma =
\begin{pmatrix}
p^{(1)}_1 &\ldots & p^{(d)}_1 \\
\vdots &\ddots & \vdots\\
p^{(1)}_l &\ldots & p^{(d)}_l \\
\end{pmatrix} 
\begin{pmatrix}
\gamma^{(1)} \\ \vdots \\ \gamma^{(d)}
\end{pmatrix}
= \begin{pmatrix}
BM_1^TB^*\\ \vdots \\ BM_l^TB^*
\end{pmatrix}
\ee
with $\gamma^{(k)} = r_k\rho^{(k)}$.
Since each $M_i$ has an eigenvalue $1$,  $\Delta_\Xi = \Delta_l$.
Hence we set $d=l$ and
\[
P = I_l; \quad \gamma^{(k)} = BM_k^TB^* 
\]
\end{proof}

We think that above condition $\Delta_\Xi = \Delta_l$ should be close to a necessary condition too.
We set $K=l$ and  think whether each block of
\[
P^{-1} \times \begin{pmatrix} 
M_1^T\\ \vdots \\ M_l^T
\end{pmatrix}
\] 
is positive or not. 
Suppose we have chosen $P$ as 
\[
\lambda I + (1-\lambda) \psi \psi^*  
\]
Here, $0 <  \lambda \leq1$ and $\psi = \frac 1 {\sqrt{l}} (1,\ldots,1)^T$. 
Set 
\[
Q=I-P = (1-\lambda)  (I-\psi\psi^*)
\]
Then, 
\be
P^{-1} = (I-Q)^{-1}  =  \sum_{i=0}^\infty Q^i  = \frac 1 \lambda  (I-\psi\psi^*)
\ee
However then this always give a non-positive block.
Indeed, the $i$-th block, rescaled, will be
\[
M_i - \frac 1 l \sum_{j=1}^l  M_j = M_i - \frac 1 l I
\]
and one of them should be non-positive.

\section*{Acknowledgements}
The authors had opportunities to meet at the LPT in Toulouse, the ICJ in Lyon, the Department of Mathematics of uOttawa, 
the TU M\"unchen and the Isaac Newton Institute in Cambridge to complete their research, and thank these institutions 
for a fruitful working environment.

BC's research was supported by NSERC discovery grants, Ontario's ERA and AIMR, Tohoku university.
MF's research was financially supported by the CHIST-ERA/BMBF project CQC.
IN's research has been supported by the ANR grants ``OSQPI'' {2011 BS01
  008 01} and ``RMTQIT''  {ANR-12-IS01-0001-01}, and by the PEPS-ICQ CNRS project ``Cogit''.

\appendix
\section{The optimization problem for random mixed unitary channels}

In this technical appendix, we provide the details of the proof for the optimization problem appearing in Theorem \ref{thm:random-mixed-unitary-channels}. Let us recall it here, for the convenience of the reader. Let $\mathcal S_C^{k-1}$ be the unit sphere of $\mathbb C^k$ and define
\begin{align}\label{function:g}
g: \, \mathcal S_\C^{k-1} \times (0,\infty) &\to \mathbb R\\
(a,x) &\mapsto (2-k)x + \sum_{i=1}^k \sqrt{x^2 + |a_i|^2w_i}
\end{align}
where $k > 2$ is an integer parameter and $(w_1, \ldots w_k)$ is a strictly positive probability vector: $w_i>0$ and $\sum_i w_i = 1$. 
Since only the absolute values $|a_i|^2$ appear in the above formula, we shall assume, without loss of generality, that the numbers 
$a_i$ are real and satisfy $\sum_i a_i^2=1$. 

In what follows, we prove the following result:
\begin{proposition}\label{prop:optimization-g}
Let $g$ be the function defined in \eqref{function:g}, but on $\mathcal S_\R^{k-1} \times (0,\infty)$ as is described above. Then,
\begin{enumerate}
\item We have the formula:
\be\label{eq:optimum}
\psi_*(\sqrt w) = \max_{J \in \mathcal J} h(J)
\ee
Here, remember that $ \psi_*(\sqrt w)= \max_{a \in \mathcal S^{k-1}} \min_{x>0} g(a,x)$
defined in \eqref{function:psi*}.
In the above formula, $\mathcal J$ is a collections of subsets of $[k] = \{1, \ldots, k\}$, defined as
\begin{equation}\label{eq:valid-J}
\mathcal J = \{ J \subset [k] \, : \, \min_{j \in J} w_j \geq \gamma | \# J -2| \}
\end{equation}
elements of which we call  \emph{valid} subsets. 
Also, the function $h(\cdot)$ is defined on $\mathcal J$ as
\be
h(J) =  \sqrt{\beta - \gamma (\#J - 2)^2}
\ee
where $\beta$ and $\gamma$ are
\begin{equation}\label{eq:beta-gamma}
\beta = \sum_{j \in J} w_j \qquad \frac{1}{\gamma} = \sum_{j \in J} \frac{1}{w_j}
\end{equation}
Note that $\mathcal J$ contains all subsets with cardinality less than or equal to $3$.
\item The function $h(\cdot)$ is well-defined on  on $2^{[k]}$  and non-decreasing with respect to the canonical partial order. 
As a result, 
if the full set $J=[k]$ is valid, i.e $\min_{j \in [k]} w_j \geq \gamma_{0}(k-2)$ with $\gamma_{0}^{-1} = \sum_{j=1}^k w_j^{-1}$, then the 
optimum is $\sqrt{1-\gamma_{0}(k-2)^2}$. In particular, when $w_i$ is the flat distribution, $w_i=1/k$, we get $a_i=1/k$ and the 
optimum is $2\sqrt{k-1}/k$.
\end{enumerate}
\end{proposition}

\begin{proof}
Let us start by giving an outline of the proof. First, we notice that the minimization problem in $x$ is convex, hence a unique minimum exists.  
Moreover, this minimum $X_a$ depends smoothly on $a$ and thus we are left with a smooth maximization 
problem in $a \in \mathcal S_{\mathbb R}^{k-1}$. 
Next, we use Lagrange multipliers to solve this problem, and we find a set of critical points indexed by subsets $J \subset [k] = \{1, \ldots, k\}$, 
where the coordinates $a_i$ are non-zero. Not all subsets $J$ yield critical points and one has to take a maximum over the set of 
valid subsets $J$ to conclude.
Finally, we show monotonic property of the function $h(\cdot)$ with respect to the partial order in $2^{[k]}$.

{\bf Step 1}:
Let us start by noticing that, at fixed $a$, the function $x \mapsto g(a,x)$ is convex, so it admits a unique minimum 
$X_a \in [0, \infty)$. Since $\frac{\partial g}{\partial x}$ is negative at $x=0$, we have $X_a > 0$ (see also Remark \ref{rem:positive-x}). 
The value $X_a$ is defined by the following implicit equation
\begin{equation}
\frac{\partial g}{\partial x}\bigg|_{x=X_a} = 0,
\end{equation}
which is equivalent to $F(a,X_a) = 0$, for
\begin{equation}\label{eq:F-a-x}
F(a,x) =\frac{\partial g}{\partial x} =  2-k + \sum_{i=1}^k \frac{x}{\sqrt{x^2+a_i^2w_i}}
\end{equation}
It follows from the implicit function theorem that the map $a \mapsto X_a$ is $C^1$ because 
\begin{equation}
\frac{\partial F}{\partial x} 
= \sum_{i=1}^k \left[ \frac{1} {(x^2+a_i^2w_i)^{1/2}} - \frac{x^2} {(x^2+a_i^2w_i)^{3/2}} \right]
= \sum_{i=1}^k \frac{a_i^2w_i} {(x^2+a_i^2w_i^2)^{1/2}}\neq 0,
\end{equation}

{\bf Step2:}
 Now we want to solve 
\begin{equation}
\max_{a \in \mathcal S_\R^{k-1}} g(a,X_a)
\end{equation}
by introducing the Lagrange multiplier functional
\begin{equation}
G(a, \lambda) = g(a,X_a) - \frac{\lambda}{2}\sum_{i=1}^k a_i^2 = (2-k)X_a +\sum_{i=1}^k \sqrt{X_a^2 + a_i^2w_i}  - \frac{\lambda}{2}\sum_{i=1}^k a_i^2
\end{equation}
The criticality condition, the normalization for $a$ and the restriction of $X_a$ translate to 
\begin{align}
\forall j \in [k], &\quad \frac{\partial X_a}{\partial a_j} \left( 2-k + \sum_{i=1}^k \frac{X_a}{\sqrt{X_a^2 + a_i^2w_i}} \right) + \frac{a_jw_j}{\sqrt{X_a^2 + a_j^2w_j}} - \lambda a_j = 0 \label{LM1}\\
&\sum_{i=1}^k a_i^2 = 1 \label{LM2} \\
&F(a,X_a) = 0 \label{LM3}
\end{align}
Below, we get candidates for the solutions for this system of equations. 

Firstly, \eqref{LM1} and \eqref{LM3} imply that 
\begin{align}
\forall j \in [k], \quad  \frac{a_jw_j}{\sqrt{X_a^2 + a_j^2w_j}} &= \lambda a_j 
\end{align}
Let us now introduce the index sets $I = \{i \, : \, a_i=0\}$ and  $J = [k] \setminus I$. 
Then, for $j \in J$ we have
\be
w_j = \lambda \sqrt{X_a^2 + a_j^2w_j}
\ee
This implies two equations: \eqref{LM3} gives
\be
0 = 2-k + \# I +  \lambda X_a \sum_{j \in J} \frac 1 {w_j}   \qquad \text{or} \qquad  \# J-2  =  \frac{\lambda X_a}{\gamma}
\ee
and, squaring the both sides yields 
\be
a_j^2 = \frac{w_j}{\lambda^2} - \frac{X_a^2}{w_j} = \frac 1{\lambda^2} \left( w_j - \frac {(\lambda X_a)^2} {w_j} \right)  
= \frac 1{\lambda^2} \left( w_j - \frac { \gamma^2 (\# J -2)^2} {w_j} \right)  
\ee

Secondly, with \eqref{LM2},  we have
\be
1 = \sum_{j\in J} a_i^2 = \frac 1 {\lambda^2} \left(\beta - \frac{(\lambda X_a)^2}{\gamma}\right) 
= \frac 1 {\lambda^2} \left(\beta - \gamma (\# J -2)^2 \right) 
\ee
This leads to 
\be
a_j ^2
= \frac{1}{\beta - \gamma (\# J -2)^2} \cdot  \left( w_j - \frac { \gamma^2 (\# J -2)^2} {w_j} \right)  \label{formula:a}
\ee
Also, 
\be
X_a =  \frac {\gamma (\# J -2)} {|\lambda|} = \frac{\gamma (\# J -2)}{\sqrt{\beta - \gamma (\# J -2)^2}}
\ee

Thirdly, for those candidates the function $g(\cdot,\cdot)$ can be simplified:
\be 
g(a,X_a) 
&=& (2- \# J) X_a + \sum_{j \in J} \sqrt{X_a^2 + a_j^2 w_j} 
= (2- \# J) X_a + \frac \beta \lambda \\
&=& - \frac{\gamma (\# J -2)^2}{\sqrt{\beta - \gamma (\# J -2)^2}} + \frac \beta {\sqrt{\beta - \gamma (\# J -2)^2}} 
= \sqrt{\beta - \gamma (\# J -2)^2}
\ee
Since this function only depends on set $J$, we redefine this function to be $h(J)$ as in the statement of theorem.

{\bf Step 3:}
So far, we get a set of candidates for solutions, but we get the actual solutions, and hence the precise set of critical points, by
thinking  positivity issues for $a_j^2$ with $j \in J$. 
The inequality between the harmonic and the arithmetic means, applied for $\{w_j\}_{j \in J}$ reads
\begin{equation}\label{eq:inequality-means}
\frac{\#J}{\sum_{j \in J} 1/w_j} \leq \frac{\sum_{j \in J} w_j}{\#J}
\end{equation}
hence we have that $(\#J)^2 \gamma \leq \beta$ for all choices of $J$. This implies that the first factor in \eqref{formula:a} is 
always strictly positive, except for $\#J =1$, when it is zero. 
Hence, looking into the second factor in \eqref{formula:a}, 
the condition that $a_j^2 \geq 0$ for all $j \in J$ 
is equivalent to the condition that $J$ is a valid subset, as in \eqref{eq:valid-J} with respect to those candidates in  \eqref{formula:a}.  
Therefore, maximizing $h(J)$ over $\mathcal J$ gives the maximum of $g(a,X_a)$ under the normalization condition on $a$.

Note that when $\#J = 1,2$ the condition for $J$ to be valid, $\min_{j \in J} w_j \geq \gamma | \# J -2|$ is trivially satisfied. 
When $\#J = 3$, the condition reads
\begin{equation}
\frac{3-2}{\min_{j \in J} w_j} \leq \frac{1}{\gamma}  = \frac{1}{w_1} + \frac{1}{w_2} + \frac{1}{w_3}
\end{equation}
which is also always fulfilled. Thus, every subset $J$ with $\#J \leq 3$ is valid. 

{\bf Step 4:}
The mean inequality \eqref{eq:inequality-means} implies also that the quantity $h(J)$ is well defined for all subsets $J \subset [k]$, 
even if $J$ is not valid. Let us show next $h$ is an increasing function of $J$ with the canonical partial order. To this end, consider 
a subset $J$, an element $s \notin J$ and put $J' = J \cup \{s\}$. With $p=\#J$, we have the following sequence of equivalent inequalities
\begin{align}
h(J)^2 &\leq h(J')^2\\
\beta - \gamma (p-2)^2 &\leq \beta' - \gamma'(p-1)^2\\
-\frac{(p-2)^2}{\frac 1 \gamma} &\leq w_s - \frac{(p-1)^2}{\frac 1 \gamma + \frac{1}{w_s}}\\
-(p-2)^2 \left( \frac 1 \gamma + \frac{1}{w_s} \right) &\leq \frac{w_s}{\gamma} \left( \frac 1 \gamma + \frac{1}{w_s} \right) - \frac{(p-1)^2}\gamma\\
\frac{1}{\gamma}\left( \frac{2(p-2)}{w_s} - \frac{1}{\gamma} \right) &\leq \frac{(p-2)^2}{w_s^2}
\end{align}
where the last one is true by the inequality: $\sqrt{ab} \leq (a+b)/2$ for $a,b>0$. In particular, we conclude that if $J=[k]$ is valid, then
\begin{equation}
\max_{J \in \mathcal J} h(J) = h([k]) = \sqrt{1-\gamma_{0}(k-2)^2}
\end{equation}
\end{proof}

As an illustration of the above result, let us consider the case $k=4$ and 
\begin{equation}
w_r=\left[ r, \frac{1-r}{3},\frac{1-r}{3},\frac{1-r}{3} \right]
\end{equation}
with $r \in (0,1/4)$. For $J=\{1,2,3,4\}$ to be valid, one must have $r \geq 2c$. By direct computation, one finds $c=r(1-r)/(8r+1)$, 
thus $J=[4]$ is valid if and only if $r \in (0, 1/10)$. We conclude that, for $r \geq 1/10$, the optimum is $h([4]) = (2r+1)/\sqrt{8r+1}$. 

Let us now study the other regime, where $r<1/10$. There are only two distinct choices for $J$ with $\# J =3$: $J_1=\{1,2,3\}$ and 
$J_2 = \{2,3,4\}$, both valid since they have cardinality $3$. One computes directly
\begin{equation}
h(J_1) = \frac{\sqrt 2}{\sqrt 3} 
\frac{2r+1}{\sqrt{5r+1}} < \frac{2\sqrt 2}{3} \sqrt{1-r} = h(J_2)
\end{equation}
We conclude that
\begin{equation}
\psi_*(\sqrt{w_r}) = \begin{cases}
\frac{2r+1}{\sqrt{8r+1}}, \quad &\text{ if } r \geq \frac{1}{10}\\
\frac{2 \sqrt 2}{3} \sqrt{1-r}, \quad &\text{ if } r < \frac{1}{10}
\end{cases}
\end{equation}

\end{document}